\documentclass{llncs}
\usepackage{etex}
\usepackage{pstricks}
\usepackage{pst-tree,pst-node}
\usepackage{amsfonts}
\usepackage{tikz}
\usepackage{array}
\usepackage{amsmath}
\usepackage{amsfonts}
\usepackage{amssymb}
\usepackage[caption=true,font=footnotesize]{subfig}
\usepackage{color}
\usepackage{verbatim}

  \newcommand{\defproblem}[3]{
  \vspace{2mm}
\noindent\fbox{
  \begin{minipage}{0.96\textwidth}
  #1\\
  {\bf{Input:}} #2  \\
  {\bf{Output:}} #3
  \end{minipage}
  }
  \vspace{2mm}
}

\usepackage[algo2e, noend, noline, boxed]{algorithm2e}
  \SetKwComment{tcm}{\{}{\}}
   
  \newenvironment{myalgorithm}[2][htbp]
  {%
    \setlength{\algomargin}{.2cm}
    \begin{center}
    \begin{minipage}{#2}
    \begin{algorithm2e}[#1]
    \small
     \let\Par=\par
       \def\par{\endgraf\vspace{.1cm}}
           \SetKw{To}{to}%
       \SetKw{Downto}{downto}%
           \SetKw{Or}{or}%
       \SetKwFor{Algo}{Algorithm}{}{}%
      \vspace{.15cm}%
   }
   {%
     \let\par=\Par\end{algorithm2e}%
     \end{minipage}%
     \end{center}%
   }

\newcommand{\WPM}{\textsc{WeightedPatternMatching}}
\newcommand{\WTM}{\textsc{WeightedTextMatching}}

\def\dd{\mathinner{.\,.}}
\newcommand{\cO}{\mathcal{O}}
\newcommand{\cco}{o}

\begin{document}
\frontmatter          % for the preliminaries
%
%\pagestyle{headings}  % switches on printing of running heads
%\addtocmark{Hamiltonian Mechanics} % additional mark in the TOC
%
\title{Fast Average-Case Pattern Matching on Weighted Sequences}

\author{Carl Barton\inst{1}
\and Chang Liu\inst{2}
\and Solon P.\ Pissis\inst{2}
}

\institute{$\!^1\ $The Blizard Institute, Barts and The London School of Medicine and Dentistry, Queen Mary University of London, UK \\
\email{c.barton@qmul.ac.uk}\\
$\!^2\ $Department of Informatics, King's College London, UK \\
\email{\{chang.2.liu,solon.pissis\}@kcl.ac.uk}
}
\maketitle

\begin{abstract}
A {\em weighted string} over an alphabet of size $\sigma$ is a string in which a set of letters may occur at each position with respective occurrence probabilities.
Weighted strings, also known as {\em position weight matrices} or {\em uncertain} sequences, naturally arise in many contexts.
In this article, we study the problem of weighted string matching with a special focus on average-case analysis.
Given a weighted {\em pattern} string $x$ of length $m$, a {\em text} string $y$ of length $n>m$, 
and a {\em cumulative weight threshold} $1/z$, defined as the minimal probability of occurrence of factors in a weighted string, 
we present an algorithm requiring average-case search time $\cco(n)$ for pattern matching for {\em weight ratio} $\frac{z}{m} <  \min\{\frac{1}{\log z},\frac{\log \sigma}{\log z (\log m + \log \log \sigma)}\}$.
For a pattern string  $x$ of length $m$, a weighted text string $y$ of length $n>m$, and a cumulative weight threshold $1/z$,
we present an algorithm requiring average-case search time $\cco(\sigma n)$ for the same weight ratio.
The importance of these results lies on the fact that these algorithms work in average-case sublinear search time in the size of the text, 
and in linear preprocessing time and space in the size of the pattern, for these ratios.
\end{abstract}
\section{Introduction}
\label{sec:intro}
A weighted string over some alphabet is a type of {\em uncertain} sequence in which {\em a set} of letters (instead of a single letter) may occur at each position with respective occurrence probabilities.
This notion was first introduced by Iliopoulos et al.~\cite{DBLP:conf/stringology/IliopoulosMPT03,DBLP:journals/jalc/IliopoulosMPT05} in 2003.
A great deal of research has been conducted ever since on weighted strings: for pattern matching~\cite{IPT04PatMatchWeightedSeq,Amir2006a,Amir2006b};
for computing various types of regularities~\cite{DBLP:conf/spire/IliopoulosPTTT04,DBLP:journals/jcb/ChristodoulakisIMPTT06,DBLP:journals/bmcbi/ZhangGI13,Barton2014}; 
for indexing~\cite{Amir2006a,Iliopoulos:2006:WST:1227505.1227512}; and for alignments~\cite{Amir2010}.

An \emph{alphabet} $\Sigma$ is a finite non-empty set of size $\sigma$, whose 
elements are called \emph{letters}. A \emph{string} on an alphabet $\Sigma$ is 
a finite, possibly empty, sequence of elements of $\Sigma$. The zero-letter 
sequence is called the \emph{empty string}, and is denoted by $\varepsilon$. The 
\emph{length} of a string $x$ is defined as the length of the sequence associated 
with the string $x$, and is denoted by $\vert x \vert$. We denote by $x[i]$, for 
all $0 \leq i < |x|$, the letter at index $i$ of $x$. Each index $i$, for all 
$0 \leq i <|x|$, is a position in $x$ when $x \neq \varepsilon$. It follows that 
the $i$-th letter of $x$ is the letter at position $i-1$ in $x$.

The {\em concatenation} of two strings $x$ and $y$ is the string of the letters of
$x$ followed by the letters of $y$; it is denoted by $xy$.
A string $x$ is a \emph{factor} of a string $y$ if there exist two strings $u$ 
and $v$, such that $y=uxv$. 
Consider the strings $x,y,u$, and $v$, such that $y=uxv$, if $u=\varepsilon$ then $x$ is a \emph{prefix} of $y$, if $v=\varepsilon$ then $x$ is a \emph{suffix} of $y$.
Let $x$ be a non-empty string and $y$ be a string, we say that there exists an \emph{occurrence} of $x$ in $y$, or more simply, that $x$ \emph{occurs in} $y$, 
when $x$ is a factor of $y$. Every occurrence of $x$ can be characterised by a 
position in $y$; thus we say that $x$ occurs at the \emph{starting position} $i$ in $y$ when $y[i \dd i + |x| - 1]=x$.

Single nucleotide polymorphisms, as well as errors introduced by wet-lab sequencing platforms
during the process of DNA sequencing, can occur in some positions of a DNA sequence. 
In some cases, these uncertainties can be accurately modelled as a {\em don't care} letter. 
However, in other cases they can be more subtly expressed, and, at each position of the sequence, a probability of occurrence can be assigned 
to each letter of the nucleotide alphabet; this process gives rise to a {\em weighted string} or a {\em position weight matrix}. For instance, consider a IUPAC-encoded~\cite{IUPAC} 
DNA sequence, where the ambiguity letter \texttt{M} occurs at some position of the sequence, representing either base \texttt{A} or  base \texttt{C}. 
This gives rise to a weighted DNA sequence, where at the corresponding position of the sequence, we can assign to each of \texttt{A} and \texttt{C} an occurrence probability of $0.5$.

A weighted string $x$ of length $n$ on an alphabet $\Sigma$ is a finite sequence of $n$ sets. 
Every $x[i]$, for all $0 \leq  i < n$, is a set of ordered pairs $(s_j, \pi_i(s_j))$, 
where $s_j \in \Sigma$ and $\pi_i(s_j)$ is the probability of having letter 
$s_j$ at position $i$. Formally, $x[i] = \{ (s_j, \pi_i(s_j)) | s_j \neq s_{\ell} \text{ for } j \neq \ell, \text{ and } \sum_{j} \pi_i(s_j) = 1 \}$.
A letter $s_j$ {\em occurs} at position $i$ of a weighted string $x$ if and only if
the {\em occurrence probability} of letter $s_j$ at position $i$, $\pi_i(s_j)$, is greater than 0.
A string $u$ of length $m$ is a {\em factor} of a weighted string if and only if it occurs at 
starting position $i$ with {\em cumulative occurrence probability} $\prod_{j=0}^{m-1}\pi_{i+j}(u[j])>0$.
Given a {\em cumulative weight threshold} $1/z \in (0,1]$, we say that factor $u$ is {\em valid}, or equivalently that factor $u$
has a valid occurrence, if it occurs at starting position $i$ and $\prod_{j=0}^{m-1}\pi_{i+j}(u[j]) \geq 1/z$.
Similarly, we say that letter $s_j$ at position $i$ is valid if $\pi_i(s_j) \geq 1/z$.
For succinctness of presentation, if $\pi_i(s_j)=1$ the set of pairs is denoted only by the letter $s_j$;
otherwise it is denoted by $[(s_{j_1}, \pi_i(s_{j_1})),\ldots,(s_{j_k}, \pi_i(s_{j_k}))]$.

In this article, we consider the following two problems.

\defproblem{\WPM}{a weighted string $x$ of length $m$, a string $y$ of length $n>m$, and a cumulative weight threshold $1/z \in (0,1]$}{all positions $i$ of $y$ where a valid factor of length $m$ of $x$ occurs}
\defproblem{\WTM}{a string $x$ of length $m$, a weighted string $y$ of length $n>m$, and a cumulative weight threshold $1/z \in (0,1]$}{all positions $i$ of $y$ where a valid factor $v$ of $y$ occurs and $x=v$}

In~\cite{Amir2006a}, Amir et al.~showed an $\cO(nz^2 \log z)$-time bound for pattern matching on weighted strings via reduction.
A direct $\cO(nz^2 \log z)$-time algorithm for solving this problem was recently presented in~\cite{DBLP:conf/cwords/BartonP15}.
Many other algorithms for solving these problems exist, but their runtime efficiency relies on the assumption of a given {\em constant} $z$ (cf.~\cite{IPT04PatMatchWeightedSeq,Iliopoulos:2006:WST:1227505.1227512}). 
In this article, we are interested in designing average-case efficient algorithms for certain realistic {\em weight ratios} $z/m$.

\paragraph*{\textbf{Our Contribution.}} 
We present efficient average-case algorithms for weighted string matching. 
Specifically, we present two new algorithms: one to solve problem \WPM~and another one to solve problem \WTM. 
Both algorithms can achieve average-case sublinear search time in the size of the text, 
and work in linear preprocessing time and space  in the size of the pattern.
Essentially, we show that they achieve these average-case search times depending on the number of
positions required, in every matching weighted factor of length $m$, to have a letter occurring with probability greater than $1-1/z$. 
We consider these conditions to be a quite realistic scenario in a wide range of applications, in particular, on molecular sequences~\cite{Yan01072005}.

\begin{comment}
mlog m/logm - m/logm 
Specifically, our contribution is twofold.
\begin{enumerate}
\item We present two new algorithms: one to solve problem \WPM~and another one to solve problem \WTM. Both algorithms can achieve average-case time and space linear in the size of the input.
Essentially, we show that the designed algorithms achieve average-case linear time if there exists at least one letter in every matching weighted factor with probability greater than $1-1/z$, that is, a quite realistic scenario.
\item We implemented these two algorithms and the fastest-known worst-case algorithm, which can solve both problems~\cite{DBLP:conf/cwords/BartonP15}, as a computer program.
Furthermore, we present an experimental study, using real and synthetic data, which demonstrates orders-of-magnitude superiority of our approach, in terms of runtime efficiency, compared to the worst-case approach.
\end{enumerate}
\end{comment}
\section{Properties and Auxiliary Data Structures}

We first start by providing a description of the key ideas on which our algorithms are based.
This informal description should help the reader understand the structure and the analysis of our algorithms later on.
These ideas are not all new but we find their combination in this context novel.

\begin{itemize}
	\item We perform a colouring stage, similar to~\cite{Amir2006a} and~\cite{Iliopoulos:2006:WST:1227505.1227512}, on the weighted pattern string. 
	This stage assigns a colour to each position of the weighted string based on the occurrence probabilities of the corresponding letters.
	We can then obtain (shown in~\cite{Amir2006a} and~\cite{Iliopoulos:2006:WST:1227505.1227512}) an upper bound on the special colour of {\em hard} positions that any valid factor can contain.
	\item With this upper bound at hand, we can make an assumption to design our average-case algorithms.
	By ignoring these hard positions, we can search fast for candidate occurrences using only the 
	{\em easy} positions (positions that are easily verified) to filter out positions of the text that could never yield a valid match.
	This assumption poses a first condition on the weight ratios $z/m$.
	%Then using the pigeonhole principle we can find a factor of $x$ of length at least $\lfloor \frac{m}{\ell+1} \rfloor$ consisting of only white and grey positions. This means we can
	%search for the exact occurrences of this factor in $y$. This stage can be trivially performed in time $\cO(n)$. These exact occurrences constitute a set of candidate occurrences of $x$ in $y$.
	\item For the average-case analysis of the designed algorithms, we need a randomness model on weighted strings.
	In the standard setting, this is trivial assuming uniform distribution of letters. For weighted strings, in order to provide a realistic scenario,
	we view a weighted string as an indeterminate string~\cite{DBLP:journals/jda/HolubSW08} and assume uniform distribution of subsets of the alphabet.
	%We can preprocess the weighted input string to compute a table $\textsf{A}$ such that for each black position $i$ and eacigma m)$-time preprocessing, we can check in $\cO(1)$-time whether 
	%a black position in $x$ matches a position in $y$ or not. With matrix $\textsf{A}$ at hand we can proceed with a fast verification step of the pre-computed candidate occurrences.
\end{itemize}
Given a weighted string $x$ of length $m$, we perform a colouring stage on $x$, similar to the one before the construction of the weighted suffix tree~\cite{Iliopoulos:2006:WST:1227505.1227512}, 
which assigns a colour to every position in $x$ according to the following scheme:
\begin{itemize}
	\item mark position $i$ \textit{black} (\texttt{b}), if {\em none} of the occurring letters at position $i$ has probability of occurrence greater than $1-1/z$.
	\item mark position $i$ \textit{grey} (\texttt{g}), if {\em one} of the occurring letters at position $i$ has probability of occurrence greater than $1-1/z$ and less than 1.
	\item mark position $i$ \textit{white} (\texttt{w}), if {\em one} of the occurring letters at position $i$ has probability of occurrence $1$.

\end{itemize}
Notice that if $z \geq 2$ , then at every white and grey position there is only one valid letter since only one
letter can have probability of occurrence greater than $1 - 1/z \geq 1/2$, whereas in a black position there maybe
several valid letters. However, if $z < 2$ there are no letters with probability of occurrence at least $1-1/z$ in a black position 
since all letters have probability of occurrence of at most $1 - 1/z > 1/z$. Therefore for the rest of this article we assume $z \geq 2$.
The colouring stage can be trivially performed in time $\cO(\sigma m)$.
\begin{lemma}[\cite{Amir2006a,Iliopoulos:2006:WST:1227505.1227512}]\label{lem:bp}
Given a weighted string $x$ and a cumulative weight threshold $1/z \in (0,1]$, any valid factor of $x$ contains at most $\ell = \lceil \log z / \log(\frac{z}{z-1}) \rceil$ black positions.
\end{lemma}
\begin{proof}
 Any letter at a black position of $x$ has probability of occurrence less than or equal to $1-1/z$.
 For any valid factor of $x$, it thus holds (in the worst case) that $(1 - 1/z)^{\ell} \geq 1/z$. Taking the logarithm at both sides yields the lemma.
 \qed
\end{proof}
The second key idea of the designed algorithms comes from the following simple fact.
This idea is used in many {\em other} pattern matching problems on strings~\cite{Navarro}.
\begin{lemma}\label{lem:php}
Given a weighted string $x$ and a cumulative weight threshold $1/z \in (0,1]$, if $\ell < m$, then 
there exists a consecutive sequence of positions of length at least $\lfloor \frac{m}{\ell+1} \rfloor$ of $x$ consisting of only white and grey positions.
\end{lemma}
\begin{proof}
 Immediate from the pigeonhole principle.\qed
\end{proof}
We can also preprocess a weighted string $x$ of length $m$ to compute a matrix $\textsf{A}[0\ldots \ell' - 1,0 \ldots \sigma - 1]$, 
such that for each black position $i$, $0 \leq i < \ell'$, and each letter $a \in \Sigma$, 
we have $\textsf{A}[i, \alpha] = 1$ if $\alpha$ occurs at the $i$th black position of $x$ and $0$ otherwise.  
After such a $\cO(\sigma m)$-time preprocessing, we can check in constant time whether 
a letter in a black position of $x$ matches a letter from another string or not. 
With matrix $\textsf{A}$ at hand, we can proceed with a fast verification step of the pre-computed candidate occurrences
using Lemma~\ref{lem:lcp} (see below).

Given two strings $u$ and $v$ in the standard setting, we say that the probability that $u[i]=v[i]$, for some position $i$ on $u$ and $v$, is given by $1/\sigma$ assuming uniform distribution of letters of the alphabet per position. 
This randomness model cannot be applied on weighted strings, where a subset of the alphabet occurs at every position of the string. 
For a given position, we rather assume a uniform distribution of all possible subsets of the alphabet, such that each letter of the subset has probability of occurrence greater than $0$;
i.e., such that the letter occurs. Under this assumption we can obtain the following lemma.
\begin{lemma}\label{lem:lcp}
Given a string $u$ and a weighted string $v$,
the expected length of the longest valid prefix of $v$ that is also a valid prefix of $u$ is bounded by six.
\end{lemma}

We start with a few definitions to reduce this problem to another one before giving a proof. An {\em indeterminate string} $x$ of length $m$ on an alphabet $\Sigma$ is a finite sequence of $m$ sets, 
such that $x[i] \subseteq \Sigma$, $x[i] \neq \emptyset$, for all $0 \le i < m$. 
If $|x[i]|= 1$, that is, $x[i]$ represents a single letter of $\Sigma$, 
we say that $x[i]$ is a {\em solid} letter.
We say that two indeterminate strings $x$ and $y$ {\em match}, denoted by $x \approx y$, if
$|x|=|y|$ and for each $i = 0,\ldots,|x|-1$, we have $x[i] \cap y[i]\neq \emptyset$.

\begin{proof}[of Lemma~\ref{lem:lcp}]
We view the weighted string $v$ as indeterminate string $v^{\prime}$ of length $|v|$ such that 
$a \in v^{\prime}[i]$ \textit{iff} $(a, \pi(a)) \in v[i]$ and $\pi(a)>0$.
Since we completely ignore letter probabilities and thereby the validity of factors---all factors are now valid---it 
suffices to show that the expected number $s>0$ of positions such that $u[0\dd s-1] \approx v^{\prime}[0\dd s-1]$ and $u[s] \notin v^{\prime}[s]$ is bounded by six.
%We further view the string $u$ as indeterminate string $u^{\prime}$ of length $|u|$ such that $u^{\prime}[i] = \{u[i]\}$. 

We consider the comparison of $u$ and $v^{\prime}$ from left to right.
We have that $\{u[i]\} \cap v^{\prime}[i]\neq \emptyset$ occurs in the following cases:
\begin{itemize}
 \item  $v^{\prime}[i]$ is solid and $\{u[i]\} = v^{\prime}[i]$
 \item  $v^{\prime}[i]$ is not solid and $u[i] \in  v^{\prime}[i]$.
\end{itemize}

Thus the total number of positive comparisons is $$\sigma\sum^{\sigma}_{j=1} \frac{j\binom{\sigma}{j}}{\sigma} = \sum^{\sigma}_{j=1} j\binom{\sigma}{j} = \sigma2^{\sigma-1}.$$
The total number of any case is $\sigma\sum^{\sigma}_{j=1} \binom{\sigma}{j} = \sigma(2^{\sigma}-1)$.
Therefore the probability $r$ of $\{u[i]\} \cap v^{\prime}[i]\neq \emptyset$ is $$r = \frac{2^{\sigma-1}}{2^{\sigma}-1} \leq 2/3,\text{ for }\sigma>1.$$
%to compute the expected value of a range of values, say 1..n, you form the sum from 1 to n of (value i) TIMES (probability of value i).
Thus the expected number $s>0$ of positions such that $u[0\dd s-1] \approx v^{\prime}[0\dd s-1]$ and $u[s] \notin v^{\prime}[s]$ can be described by the summation of infinite terms
$$s = r + 2r^2 + \cdots = \sum_{k=1}^{\infty} k r^k,$$
which is bounded by $r/(1-r)^2 = 6$, for $r \leq 2/3$. This concludes the proof.\qed
\end{proof}

In~\cite{Amir2006a} it was shown that $\ell = \cO( z \log z )$. Here we refine this to an exact bound which is useful later on in the analysis of the algorithms. 

\begin{lemma}\label{lem:ell}
Let $z \geq 2$. Then $\ell \leq z \log z$.
\end{lemma}
\begin{proof}
 By Lemma~\ref{lem:bp} we know that $\ell=\lceil\frac{\log z}{\log(\frac{z}{z-1})}\rceil$. For $z>1$, we must show that:
$$\ell=\lceil\frac{\log z}{\log(\frac{z}{z-1})}\rceil=\lceil\frac{\log z}{\log(z) - \log(z-1)}\rceil \leq z \log z. \text{ Or equivalently that:}$$
$$\frac{\log z( z \log z - z \log (z-1) -1)}{\log z - \log (z-1)} >0.$$
\noindent Clearly the above is true if and only if: $z \log z - z \log (z-1) -1 > 0$. There is a discontinuity at $z=1$; 
after this it is {\em always} positive and the following holds: $$\lim_{z \to \infty} z \log z - z \log (z-1) -1 = 0.$$ \qed
\end{proof}

\section{Algorithms}
\label{sec:algo}
We are now in a position to present Algorithm $\textsf{WPM}$ to solve problem~\WPM.
In this problem, we are given a weighted string $x$ of length $m$, a string $y$ of length $n>m$, and a cumulative weight threshold $1/z \in (0,1]$,
and we need to find all positions $i$ of $y$ where a valid factor of length $m$ of $x$ occurs.
  \begin{myalgorithm}[H]{11 cm}
  \Algo{$\textsf{WPM}(x,m,y,n,1/z, \Sigma)$}{
      Perform the colouring stage on $x$\;
      Find the number $\ell'$ of black positions in $x$\;
      $\sigma \leftarrow |\Sigma|$\;
      Compute $\textsf{A}[0 \ldots \ell' - 1,0 \ldots \sigma - 1]$ of $x$\;
      \If{$\ell' < m$}{
        Find the longest factor $f$ in $x$ with no black positions\;
      }
      \Else
      {
	\Return{\texttt{FAIL}}\;
      }
      Search for $f$ in $y$\;
      \ForEach{occurrence of $f$ in $y$}{
      Check if $f$ is extensible to the left using $\textsf{A}$\;
      Check if $f$ is extensible to the right using $\textsf{A}$\;
      \If{the length of extension is at least $m$}{
         Verify the validity of the factor of $x$ and report the position\;
        }
       }
  }
  \end{myalgorithm}
\begin{theorem}\label{the:main1}
Algorithm $\textsf{WPM}$ correctly solves problem \WPM, achieving 
average-case search time $\cco(n)$, if and only if $$\frac{z}{m} <  \min\Bigg\{\frac{1}{\log z},\frac{\log \sigma}{\log z (\log m + \log \log \sigma)}\Bigg\}.$$
%$$\frac{z}{m} < \frac{\log \sigma}{\sigma \log z \log m}.$$
Algorithm $\textsf{WPM}$ requires preprocessing time and space $\cO(\sigma m)$.
\end{theorem}
\begin{proof} 
Let $\ell < m$. By Lemma~\ref{lem:ell}, the maximal value of $\ell$ is $z \log z$, hence when $z \log z<m$ we obtain condition $\frac{z}{m} < \frac{1}{\log z}$.
In this case, by Lemma~\ref{lem:php} and the correctness of the average-case time-optimal searching algorithm~\cite{Yao},
all positions $i$ of $y$ where a valid factor of length $m$ of $x$ occurs must be na\"{i}vely verified in the \textbf{for} loop of the algorithm; therefore the algorithm is correct.

The colouring stage on string $x$ can be trivially performed in time $\cO(\sigma m)$. The preprocessing time and space cost for array $\textsf{A}$ of $x$ is $\cO(\sigma m)$.
Assuming $\ell < m$, by Lemma~\ref{lem:php}, the minimum factor length in $x$ with no black positions is at least $\lfloor \frac{m}{\ell+1} \rfloor$.
This factor $f$ is viewed as a standard string obtained by choosing in all grey and black positions the most probable letter.
Searching $f$ in $y$ can be performed in average-case time $\cO(\frac{n\log(m/\ell)}{m/\ell})$~\cite{Yao}.
The preprocessing time and space for searching is $\cO(m)$. The number of expected occurrences is $n / \sigma^{\lfloor m/ (\ell+1) \rfloor}$. 

Let us denote the cost of verification per occurrence by $\textsf{VER}(m,z)$. Algorithm $\textsf{WPM}$ achieves average-case search time $\cO(\frac{n\log(m/\ell)}{m/\ell})=\cco(n)$ \textit{if and only if}
$$\frac{\textsf{VER}(m,z) n}{\sigma^{\lfloor m/ (\ell+1) \rfloor}} \leq c \frac{n (\ell+1)\log_\sigma\frac{m}{\ell+1}}{m}$$ for some fixed constant $c$. 
That is, the total average-case verification cost is no more than the average-case searching cost. We take $\sigma$-based logarithms to obtain

$$\log_{\sigma}\textsf{VER}(m,z) +  \log_{\sigma}m - \log_{\sigma}c - \log_{\sigma}(\ell+1) - \log_{\sigma}\log_\sigma\frac{m}{\ell+1} \leq  m / (\ell + 1).$$

By Lemma~\ref{lem:lcp} and using array $\textsf{A}$ of $x$, we know it is possible to pick $c=\textsf{VER}(m,z)$ to obtain a maximum value for $\ell$, that is
%\[
%\frac{\log m}{\log (\ell+1) + \log \log\frac{m}{\ell+1}} \leq  m / (\ell + 1) \text{, which gives } \ell <  m \Bigg( \frac{\log (\ell+1) + \log \log\frac{m}{\ell+1}}{\log m}\Bigg).
%\]
\[
\log_\sigma \frac{m}{(\ell+1)\log_\sigma\frac{m}{\ell+1}} \leq  m / (\ell + 1) \text{, which gives } \ell <  m \Bigg( \frac{1}{\log_\sigma \frac{m}{(\ell+1)\log_\sigma\frac{m}{\ell+1}}}\Bigg).
\]

%\noindent From the previous condition on $\ell$, $\ell < m$, we obtain

%\[
%\log_\sigma \frac{m}{(\ell+1)\log_\sigma\frac{m}{\ell+1}} \geq 1 \text{, which gives } \frac{m}{(\ell+1)\log_\sigma\frac{m}{\ell+1}} \geq \sigma.
%\]

%\noindent This, in turn, gives a stricter condition on $\ell$,
%$$\frac{z}{m} \leq \frac{1}{\sigma \log z \log \frac{m}{\ell +1}}$$
%$$\ell < \frac{m}{\sigma\log_\sigma\frac{m}{\ell+1}} = \frac{m\log \sigma}{\sigma\log\frac{m}{\ell+1}}.$$

%\textsf{Solon: Getting rid of the 1 in the above does not work right? Here a 1 has been added, making the upper bound less strict, we can only make it more strict. The change from $\leq$ to $<$ is not sufficient as $\ell$ is not required to be an integer?}

%\noindent We now assume a stricter condition $\ell < m/\sigma$ which gives $\log_\sigma \frac{m}{\ell+1} \geq 1$.
%\noindent By Lemma~\ref{lem:ell} we obtain the stricter condition

\noindent Therefore we get the following second condition, simplified slightly for comprehension,

$$\ell <  m\Bigg(\frac{\log \sigma}{\log m - \log (\ell +1 ) + \log \log \sigma - \log \log m + \log \log (\ell +1)}\Bigg).$$

%$$\frac{z}{m} <  \frac{\log \sigma}{\log z (\log m - \log (\ell +1 ) + \log \log \sigma - \log \log m + \log \log (\ell +1))}.$$

\noindent By Lemma~\ref{lem:ell}, $\ell \leq z\log z$, and by the previous condition on $\ell$, $\ell < m$, we can further simplify the second condition to

$$\frac{z}{m} <  \frac{\log \sigma}{\log z (\log m + \log \log \sigma)},$$

%$$\frac{z}{m} < \frac{\log \sigma}{\sigma \log z \log m},$$

% (HERE WE REPLACE $\ell$ by $z\log z$) we obtain 

%$$\frac{z}{m} <  m \Bigg( \frac{\log (z\log z + 1) + \log \log\frac{m}{z\log z+1}}{\log z\log m}\Bigg)$$
%From here we get also the second condition $\log \log\frac{m}{z\log z+1} \geq 0$, that gives $\frac{z}{m} < \frac{1}{2 \log z}$.
%Given the condition that $\log \log\frac{m}{z\log z+1}$ is non-negative and the fact that $\log (z\log z+1) > \log z$, we obtain the third stricter condition
%$$\frac{z}{m} < \frac{1}{\log m},$$
\noindent and this concludes the proof. \qed
\end{proof}
%
\begin{comment}
\begin{corollary}\label{coro:main1}
Algorithm $\textsf{WPM}$ achieves average-case search time $\cco(n)$, if and only if there exist at least $\lceil \frac{\sigma m\log m - m \log \sigma}{\sigma\log m} \rceil$ 
positions in $x$, where a letter occurs with probability greater than $1-1/z$.
\end{corollary}
%
\begin{proof}
 %If $\log z > \log m$, by Theorem~\ref{the:main1}, Algorithm $\textsf{WPM}$ achieves average-case search time $\cco(n)$, if and only if $z/m < 1/\log z$, which, by Lemma~\ref{lem:ell}, gives $\ell < m$.
 %This means that there can be at most $m-1$ black positions in $x$.
 %If $\log z \leq \log m$, 
 %By Theorem~\ref{the:main1}, Algorithm $\textsf{WPM}$ achieves average-case search time $\cco(n)$, if and only if $z/m < 1/\log m$. 
 %By Lemma~\ref{lem:ell}, we know that $z \leq z\log z \leq \ell$, for all $z \geq 2$, which gives $\ell < m \log z/\log m$.
 %This means that there can be at most $\lfloor m/\log m \rfloor$ black positions in $x$. \qed
 By Theorem~\ref{the:main1}, Algorithm $\textsf{WPM}$ achieves average-case search time $\cco(n)$, if and only if $\frac{z}{m} < \frac{\log \sigma}{\sigma\log z\log m}$. 
 By Lemma~\ref{lem:ell}, we know that $z\log z \leq \ell$, for all $z \geq 2$, which gives $\ell < \frac{m\log \sigma}{\sigma\log m}$.
 This means that there can be at most $\lfloor \frac{m\log \sigma}{\sigma\log m} \rfloor$ black positions in $x$. \qed
\end{proof}
\end{comment}
\begin{corollary}\label{coro:main1a}
Let $\sigma > m \log \sigma$. Algorithm $\textsf{WPM}$ achieves average-case search time $\cco(n)$, if and only if there exist at least $1$ 
position in $x$, where a letter occurs with probability greater than $1-1/z$.
\end{corollary}
\begin{proof}
 For $\sigma > m \log \sigma$, we know $\log\sigma > \log m + \log \log \sigma$.
 By Theorem~\ref{the:main1}, Algorithm $\textsf{WPM}$ achieves average-case search time $\cco(n)$, if and only if $\frac{z}{m} < \frac{1}{\log z}$. 
 By Lemma~\ref{lem:ell}, we know that $\ell \leq z\log z$, for all $z \geq 2$, which gives $\ell < m$.
 This means that there can be at most $m - 1$ black positions in $x$. \qed
\end{proof}
We can similarly obtain the following complementary corollary.
\begin{corollary}\label{coro:main1b}
Let $\sigma \leq m \log \sigma$. Algorithm $\textsf{WPM}$ achieves average-case search time $\cco(n)$, if and only if there exist at least $\lceil \frac{m(\log m + \log\log \sigma) - m\log \sigma}{\log m + \log\log \sigma} \rceil$ 
positions in $x$, where a letter occurs with probability greater than $1-1/z$.
\end{corollary}
\begin{proof}
 For $\sigma \leq m \log \sigma$, we know $\log\sigma \leq \log m + \log \log \sigma$.
 By Theorem~\ref{the:main1}, Algorithm $\textsf{WPM}$ achieves average-case search time $\cco(n)$, if and only if $\frac{z}{m} < \frac{\log \sigma}{\log z (\log m + \log \log \sigma)}$. 
 By Lemma~\ref{lem:ell}, we know that $\ell \leq z\log z$, for all $z \geq 2$, which gives $\ell <  \frac{m \log \sigma}{\log m + \log \log \sigma}$.
 This means that there can be at most $\lfloor \frac{m \log \sigma}{\log m + \log \log \sigma} \rfloor$ black positions in $x$. \qed
\end{proof}
We next present Algorithm $\textsf{WTM}$ to solve problem~\WTM.
In this problem, we are given a string $x$ of length $m$, a weighted string $y$ of length $n>m$, and a cumulative weight threshold $1/z \in (0,1]$,
and we need to find all positions $i$ of $y$ where a valid factor $v$ of length $m$ of $y$ occurs and $v=x$.

For the searching stage of Algorithm $\textsf{WTM}$, we view weighted string $y$ as the string $y'$ according to the following scheme:
\begin{itemize}
 \item if position $i$ is white, then $y'[i]=\alpha$, where $(\alpha, \pi(\alpha)) \in y[i]$, $\alpha \in \Sigma$, and $\pi(\alpha)=1$.
 \item if position $i$ is grey, then $y'[i]=\alpha$, where $(\alpha, \pi(\alpha)) \in y[i]$, $\alpha \in \Sigma$, and $\pi(\alpha) > 1 - 1/z$.
 \item if position $i$ is black, then $y'[i]=\lambda$, where $\lambda \notin \Sigma$.
\end{itemize}

Intuitively, while searching, we assign for every black position of $y$ to $y'$ a letter $\lambda$ that is not in $\Sigma$.
This in turn implies that writing a position on string $y'$ requires time $\cO(\sigma)$.
In~\cite{CrochemoreCGLPR99}, it was shown that searching for a set of patterns of total length $M$ in a text of length $n$
requires average-case search time $\cO(\frac{n \log m}m)$ if the length $m$ of the shortest pattern in the set is polynomial in $M$.
Hence we can do searching for a set of patterns in $y'$ in average-case time $\cO(\frac{\sigma n \log m}{m})$.
The additional factor $\sigma$ is due to the cost of writing and, subsequently, reading a letter of string $y'$.
Notice that string $y'$ is implicit, we never actually construct it; and we never perform a colouring stage on $y$ as these would require time $\cO(\sigma n)$.

  \begin{myalgorithm}[H]{11 cm}
  \Algo{$\textsf{WTM}(x,m,y,n,1/z, \Sigma)$}{   
      $\ell \leftarrow \lceil \log z / \log(\frac{z}{z-1}) \rceil$\;
      \If{$\ell < m$}{
         Partition $x$ in $\ell+1$ non-overlapping fragments $f_0,f_1,\ldots,f_{\ell}$\;
         Each fragment is of length at most $\lceil \frac{m}{\ell+1} \rceil$ and at least $\lfloor \frac{m}{\ell+1} \rfloor$\;
      }
      \Else
      {
	\Return{\texttt{FAIL}}\;
      }
      Search for $f_0,f_1,\ldots,f_{\ell}$ by considering string $y'$\;
      \ForEach{occurrence of $f \in \{f_0,f_1,\ldots,f_{\ell}\}$ in $y'$}{
      Check if $f$ is extensible to the left na\"{i}vely\;
      Check if $f$ is extensible to the right na\"{i}vely\;
      \If{the length of extension is at least $m$}{
         Verify the validity of the factor of $y$ and report the position\;
        }
       }
  }
  \end{myalgorithm}
\begin{theorem}\label{the:main2}
Algorithm $\textsf{WTM}$ correctly solves problem \WTM, achieving 
average-case search time $\cco(\sigma n)$, if and only $$\frac{z}{m} <  \min\Bigg\{\frac{1}{\log z},\frac{\log \sigma}{\log z (\log m + \log \log \sigma)}\Bigg\}.$$
Algorithm $\textsf{WTM}$ requires preprocessing time and space $\cO(m)$.
\end{theorem}
\begin{proof}
Let $\ell < m$. By Lemma~\ref{lem:ell}, the maximal value of $\ell$ is $z \log z$, hence when $z \log z < m$ we obtain condition $\frac{z}{m} < \frac{1}{\log z}$.
In this case, by Lemma~\ref{lem:php} and the correctness of the average-case time-optimal searching algorithm~\cite{CrochemoreCGLPR99}, 
all positions $i$ of $y$ where a valid factor $v$ of $y$ occurs, such that $x = v$, must be na\"{i}vely verified in the \textbf{for} loop of the algorithm; therefore the algorithm is correct.

Assuming $\ell < m$, by Lemma~\ref{lem:php}, the minimum factor length with no black positions in any factor of length $m$ of $y$ is at least $\lfloor \frac{m}{\ell+1} \rfloor$.
Searching stage of the $\ell + 1$ fragments of $x$ (of length at most $\lceil \frac{m}{\ell+1} \rceil$ and at least $\lfloor \frac{m}{\ell+1} \rfloor$) 
in string $y'$ can be performed in average-case search time $\cO(\frac{\sigma n \log (m/\ell)}{m/\ell})$ (see~\cite{CrochemoreCGLPR99} and the analysis above).
The preprocessing time and space for searching is $\cO(m)$. The number of expected occurrences is $n(\ell+1) / \sigma^{\lfloor m/ (\ell+1) \rfloor}$. 

Let us denote the cost of verification per occurrence by $\sigma\textsf{VER}(m,z)$.
The additional factor $\sigma$ is due to the cost of reading a letter of string $y$ na\"{i}vely.
Algorithm $\textsf{WTM}$ achieves average-case search time $\cO(\frac{\sigma n\log(m/\ell)}{m/\ell})=\cco(\sigma n)$ \textit{if and only if}
$$\frac{(\ell + 1)\sigma\textsf{VER}(m,z) n}{\sigma^{\lfloor m/ (\ell+1) \rfloor}} \leq  c \frac{\sigma n (\ell+1)\log_{\sigma}\frac{m}{\ell+1}}{m}$$ for some fixed constant $c$. 
That is, the total average-case verification cost is no more than the average-case searching cost.
We take $\sigma$-based logarithms to obtain

$$\log_{\sigma}\textsf{VER}(m,z) + \log_{\sigma}m - \log_{\sigma}c - \log_{\sigma}\log_{\sigma} \frac{m}{\ell+1} \leq  m / (\ell+1).$$

\noindent By Lemma~\ref{lem:lcp}, we know it is possible to pick $c=\textsf{VER}(m,z)$ to obtain 

$$\log_{\sigma}m - \log_{\sigma}\log_{\sigma} \frac{m}{\ell+1} \leq  m / (\ell+1).$$

\noindent This gives a maximum value for $\ell$, that is
\[
\log_{\sigma}\frac{m}{\log_{\sigma}\frac{m}{\ell+1}} \leq  m / (\ell + 1) \text{, which gives } \ell <  m\Bigg(\frac{1}{\log_{\sigma}\frac{m}{\log_{\sigma}\frac{m}{\ell+1}}}\Bigg).
\]

%\noindent From here we obtain
%\[
%\log_{\sigma}\frac{m}{\log_{\sigma}\frac{m}{\ell+1}} \geq 1 \text{, which gives } \frac{m}{\log_\sigma\frac{m}{\ell+1}} \geq \sigma \text{, which in turn gives } \frac{\sigma^m}{\frac{m}{\ell +1}} \geq \sigma^{\sigma}.
%\]

%\noindent From some simple rearrangement we then get the following lower bound on $\ell$.

%$$\ell \geq m \sigma^{\sigma -m } -1$$

\noindent Therefore we get the following second condition, simplified slightly for comprehension,

%$$\frac{z}{m} \geq \frac{\sigma^{\sigma -m } -1}{\log z}$$

%\noindent From here we get also a third condition
%\[
%\ell <  \frac{m}{\log_\sigma m}.
%\]

%$$\frac{\sigma^{\sigma -m } -1}{\log z} \leq \frac{z}{m} <  \frac{\log \sigma}{\log z (\log m + \log \log \sigma - \log \log m)},$$

$$\ell <  m\Bigg(\frac{\log \sigma}{\log m + \log \log \sigma - \log \log m + \log \log (\ell +1)}\Bigg).$$

%$$\frac{z}{m} <  \frac{\log \sigma}{\log z (\log m + \log \log \sigma - \log \log m + \log \log (\ell +1))}.$$

\noindent By Lemma~\ref{lem:ell}, $\ell \leq z\log z$, and the previous condition on $\ell$, $\ell < m$, we can further simplify the second condition to

$$\frac{z}{m} <  \frac{\log \sigma}{\log z (\log m + \log \log \sigma)},$$

%$$\frac{\sigma^{\sigma -m } -1}{\log z} \leq \frac{z}{m} <  \frac{\log \sigma}{\log z (\log m + \log \log \sigma - \log \log m)},$$ 

\noindent and this concludes the proof.\qed

%\noindent Asymptotically on $m$ the upper bound we have here is more restrictive than the simple $\frac{1}{\log z}$ bound at the point that $(\log m + \log \log \sigma - \log \log m) > \log \sigma$ and this concludes the proof.\qed

%$$\frac{z}{m} <  \frac{\log \log\frac{m}{z\log z + 1}+1}{\log z\log m}.$$ 

%From here we get also the second condition $\log \log\frac{m}{z\log z+1} \geq 0$, that gives $\frac{z}{m} < \frac{1}{2 \log z}$.
%Given the condition that $\log \log\frac{m}{z\log z + 1}$ is non-negative, we obtain the third stricter condition
%$$\frac{z}{m} < \frac{1}{\log z\log m},$$
\end{proof}

%Note that the lower bound we achieve here is not very restrictive. 
%Assuming that $\sigma = \cco(m)$, it can easily be verified that, asymptotically on $m$, 
%there exists a point where the $m$ factor in the exponent overcomes the $\sigma$ factor and this lower bound can be made arbitrarily small. Additionally the factor $\log m + \log \log \sigma - \log \log m$ becomes $\log m$ when $\sigma = \cco(m)$.

Similar to the previous problem we obtain the following complementary corollaries.

\begin{corollary}\label{coro:main2a}
Let $\sigma > m \log \sigma$. Algorithm $\textsf{WTM}$ achieves average-case search time $\cco(\sigma n)$, if and only if there exist at least $1$ 
position in any factor of length $m$ of $y$, where a letter occurs with probability greater than $1-1/z$.
\end{corollary}
\begin{corollary}\label{coro:main2b}
Let $\sigma \leq m \log \sigma$. Algorithm $\textsf{WTM}$ achieves average-case search time $\cco(\sigma n)$, 
if and only if there exist at least $\lceil \frac{m(\log m + \log\log \sigma) - m\log \sigma}{\log m + \log\log \sigma} \rceil$ 
positions in any factor of length $m$ of $y$, where a letter occurs with probability greater than $1-1/z$.
\end{corollary}
\begin{table}[t]
\begin{center}
\begin{tabular}{|c|c|c|c|c|c|c|}
\hline Algo. & Weighted & Weight Ratio & Space & Preprocess. & Search \\ 
                 & String & ($z/m$) & & Time & Time  \\ \hline
$\textsf{WPM}$ & Pattern  & $\min\Bigg\{\frac{1}{\log z},\frac{\log \sigma}{\log z (\log m + \log \log \sigma)}\Bigg\}$              & $\cO(\sigma m)$ & $\cO(\sigma m)$ & $\cco(n)$\\ \hline
$\textsf{WTM}$ & Text     & $\min\Bigg\{\frac{1}{\log z},\frac{\log \sigma}{\log z (\log m + \log \log \sigma)}\Bigg\}$& $\cO(m)$        & $\cO(m)$        & $\cco(\sigma n)$\\ \hline 
\end{tabular}
\caption{Average-case algorithms for weighted string matching}
\label{tab:comparison}
\end{center}
\end{table}

\noindent In Table~\ref{tab:comparison}, we summarise the results presented in this section and conclude with the following remarks.

\begin{remark}
 In most typical applications we have that $\sigma \leq m \log \sigma$, hence, Corollaries~\ref{coro:main1b} and~\ref{coro:main2b} apply 
 rather than Corollaries~\ref{coro:main1a} and~\ref{coro:main2a}, respectively.
\end{remark}

A typical example application of Corollaries~\ref{coro:main1b} and~\ref{coro:main2b} is provided below.

\begin{example}
 Let $m = 32$ and $\sigma = 4$. Algorithms \textsf{WPM} and \textsf{WTM} achieve search time sublinear, on average, in the size of the text,
 if every matching weighted factor of length $m$ has at most 10 black (uncertain) positions.
\end{example}

\begin{remark}[Worst-Case]
 If Algorithm \textsf{WPM} returns \texttt{FAIL}, we can na\"{i}vely verify all positions of $y$ in worst-case time $\cO(\sigma m + nz^2\log z)$ which matches the worst-case search time
 of the algorithm. Similarly the worst-case search time for Algorithm \textsf{WTM} is $\cO(\sigma n + nz^2\log z)$.
\end{remark}

%\section{Experimental Results}
\section{Final Remarks}
Amir et al., in their seminal work~\cite{Amir2006a}, showed an $\cO(nz^2 \log z)$-time bound for pattern matching on weighted strings via reduction.
A direct $\cO(nz^2 \log z)$-time algorithm for solving this problem was recently presented in~\cite{DBLP:conf/cwords/BartonP15}.
This algorithm is efficient when $z$ is constant or logarithmic in the length of the pattern.

In this article, we designed algorithms for weighted string matching that can achieve search time sublinear, on average, in the size of the text.
We also showed exact bounds on the weighted ratio $z/m$ in order to achieve this time complexity.
Both the cost of preprocessing time and that of space requirements are linear in the size of the pattern. 

As a by-product, we also showed upper bounds on when the designed algorithms achieve these average-case search times. 
These bounds depend on the number of positions required, in every matching weighted factor of length $m$, 
to have a letter occurring with probability greater than $1-1/z$ (Corollaries~\ref{coro:main1a}--\ref{coro:main2b}).
We consider these conditions to be a quite realistic scenario in a wide range of applications. 
A particular example application is for finding IUPAC-encoded nucleotide or peptide sequences such as \textit{cis}-elements in nucleotide sequences 
or small domains and motifs in protein sequences~\cite{Yan01072005}.

Our immediate target is twofold.
\begin{enumerate}
 \item From a theoretical point of view, we are planning to extend our approach here to address the problem of weighted string matching when both
 the pattern and the text are weighted strings. %It would also be interesting to achieve average-case optimality.
 \item From a practical point of view, we are planning to implement the presented average-case algorithms and the worst-case algorithm presented in~\cite{DBLP:conf/cwords/BartonP15},
 and evaluate them with real and synthetic data. %For the aforementioned weight ratios, we anticipate order-of-magnitude superiority, in terms of time, of the average-case algorithms.
\end{enumerate}

%%%%%%%%%%%%%%%%%%%%%%%%%%%%%%%%%%%%%%%%%%%%%%%%%%%%%%%%%%%%%%%%%%%%%%%%%%%%%%%%%%%%%%%%%%%%%%%%%%%%%%
\section*{Acknowledgements}
%%%%%%%%%%%%%%%%%%%%%%%%%%%%%%%%%%%%%%%%%%%%%%%%%%%%%%%%%%%%%%%%%%%%%%%%%%%%%%%%%%%%%%%%%%%%%%%%%%%%%%
%
We warmly thank G.~Fici from the University of Palermo for useful discussions. 

\bibliographystyle{splncs03}
\bibliography{references}

\end{document}